\documentclass[submission,copyright,creativecommons]{eptcs}
\providecommand{\event}{GandALF 2024} 

\usepackage{iftex}

\ifpdf
  \usepackage{underscore}         
  \usepackage[T1]{fontenc}        
\else
  \usepackage{breakurl}           
\fi

\usepackage{graphicx}
\usepackage{amssymb,amsmath}
\usepackage{xspace}
\usepackage{tikz-cd}
\usepackage{enumitem}
\usepackage{comment}

\usepackage{booktabs,arydshln,tabularray,framed,dashbox}%
\usepackage{subcaption}
\usepackage{newtx}

\usepackage{apxproof}
\usepackage[pdf,all,cmtip]{xy}
\allowdisplaybreaks

\usepackage{algorithm,algorithmicx,algpseudocodex}
\algrenewcommand\algorithmicrequire{\textbf{Input:}}
\algrenewcommand\algorithmicensure{\textbf{Output:}}
\algnewcommand{\Initialize}[1]{\State\textbf{Initialize:} #1}
\algnewcommand\True{\textbf{true}\xspace}
\algnewcommand\False{\textbf{false}\xspace}
\usepackage{multicol}

\newtheorem{theorem}{Theorem}
\newtheorem{corollary}[theorem]{Corollary}%
\newtheorem{definition}[theorem]{Definition}%
\newtheorem{lemma}[theorem]{Lemma}%
\newtheoremrep{lemma}[theorem]{Lemma}%
\newtheorem{proposition}[theorem]{Proposition}%
\newtheoremrep{proposition}[theorem]{Proposition}%
\newtheorem{remark}{Remark}

\newcommand{\algand}{\textbf{and}\xspace}

\newcommand{\ab}{\ensuremath{A}\xspace}
\newcommand{\ag}{\text{\normalfont\textsf{A}}\xspace}

\newcommand{\diamondplus}{\rotatebox[origin=c]{45}{$\boxtimes$}}
\newcommand{\diamondplussc}{\rotatebox[origin=c]{45}{\scriptsize$\boxtimes$}}
\newcommand{\diamondminus}{\rotatebox[origin=c]{135}{$\boxslash$}}

\newcommand{\Lra}{\Leftrightarrow}
\newcommand{\lra}{\leftrightarrow}

\renewcommand{\phi}{\varphi}
\newcommand{\pr}{\text{\normalfont\textsf{P}}\xspace}
\newcommand{\Ra}{\Rightarrow}
\newcommand{\ra}{\rightarrow}
\newcommand{\sk}{\text{\normalfont\textsf{S}}\xspace}

\newcommand{\mbN}{\mathbb{N}}

\newcommand{\lang}{\ensuremath{\mathcal{L}}\xspace}

\newcommand{\langdm}{\ensuremath{\mathcal{L}_{DF}}\xspace}

\newcommand{\langdf}{\langdm\xspace}

\renewcommand{\l}{\text{\normalfont L}\xspace}

\title{Epistemic Skills\\{\Large Logical Dynamics of Knowing and Forgetting}}
\author{Xiaolong Liang
\institute{School of Philosophy\\Shanxi University\\
	Taiyuan, Shanxi, P.R. China}
\email{lianghillon@gmail.com}
\and
Y\`{i} N. W\'{a}ng\thanks{Corresponding author.}
\institute{Department of Philosophy (Zhuhai)\\Sun Yat-sen University\\
	Zhuhai, Guangdong, P.R. China}
\email{ynw@xixilogic.org}
}

\newcommand{\titlerunning}{Epistemic Skills: Logical Dynamics of Knowing and Forgetting}
\newcommand{\authorrunning}{Liang and W\'{a}ng}

\hypersetup{
  bookmarksnumbered,
  pdftitle    = {\titlerunning},
  pdfauthor   = {\authorrunning},
  pdfsubject  = {EPTCS},               
  pdfkeywords = {epistemic skills, knowing, forgetting, weighted models} 
}

\begin{document}
\maketitle

\begin{abstract}
We present a type of epistemic logics that encapsulates both the dynamics of acquiring knowledge (knowing) and losing information (forgetting), alongside the integration of group knowledge concepts. Our approach is underpinned by a system of weighted models, which introduces an ``epistemic skills'' metric to effectively represent the epistemic abilities associated with knowledge update. In this framework, the acquisition of knowledge is modeled as a result of upskilling, whereas forgetting is by downskilling. Additionally, our framework allows us to explore the concept of ``knowability,'' which can be defined as the potential to acquire knowledge through upskilling, and facilitates a nuanced understanding of the distinctions between epistemic de re and de dicto expressions. We study the computational complexity of model checking problems for these logics, providing insights into both the theoretical underpinnings and practical implications of our approach.
\end{abstract}

\section{Introduction}

The study of epistemic logic has become a prolific area within applied modal logic, since its inception as a formal methodology in epistemology \cite{Wright1951,Hintikka1962}, and its subsequent application in computer science \cite{FHMV1995,MvdH1995}. A longstanding focus of this field has been to elucidate various forms of group knowledge, with mutual knowledge (what everyone knows), common knowledge, and distributed knowledge being particularly well-known concepts.

On top of this has been an exploration of actions that bring about changes in knowledge, such as the effect of public announcements. This inquiry has given rise to the subfield of dynamic epistemic logic \cite{vDvdHK2008}, a discipline that incorporates update modalities into its language to depict knowledge updates, with Public Announcement Logic \cite{Plaza1989} and Action Model Logic \cite{BMS1998} being popular approaches (the first can be viewed as a specific instance of the broader framework of the latter). Extensions of Public Announcement Logic that incorporate the concept of knowability have then garnered significant interest \cite{BBDHHL2008,ABDS2010}. These extensions delve into the nuanced understanding of what it means for something to be knowable in a dynamic informational context.

The literature presents a diverse array of approaches to model the phenomenon of forgetting within the frameworks of both classical and non-classical logics. Among these approaches, two prominent categories emerge: syntactical and semantical strategies for representing knowledge contraction. Syntactical strategies, such as those delineated by the AGM paradigm \cite{AGM1985}, typically involve the removal of formulas from an agent's knowledge base, akin to belief contraction. On the other hand, semantical strategies focus on the modification of the interpretation of knowledge. This can include various methods such as erasing the truth values assigned to atomic propositions \cite{LR1994,LLM2003,DHLM2009,ZZ2009}. Another semantical method involves updating the set of propositions that an agent is aware of \cite{FH1988}.

In this study, we propose a unified logical framework designed to model group knowledge, processes that may lead to knowledge update and epistemic necessity and possibility. Our approach is based on weighted modal logic \cite{LM2014,HLMP2018}. We extend this foundation by introducing the concept of \emph{epistemic skills}, utilizing weights assigned to the edges in our model to represent the specific skills required to distinguish between pairs of possible worlds. This differentiation introduces a measure of similarity, aligning our work with recent developments in epistemic logic that employ concepts of similarity or distance\cite{NT2015,DLW2021}.

Traditionally understood notions such as mutual and common knowledge are preserved in their classical interpretations within our framework. Additionally, we incorporate distributed and field knowledge seamlessly. Our model explicitly defines the skill set each agent possesses, and by leveraging update modalities, we model the acquisition, loss, revision of knowledge as results of upskilling, downskilling and reskilling, respectively.
By focusing on operations that modify one's skills, we broaden our analysis to include the concepts of knowability and forgetfulness. In keeping with the perspectives suggested in \cite{BBDHHL2008}, our guiding principles are: \emph{the knowable is what becomes known after upskilling,} and conversely, \emph{the forgettable is what becomes unknown upon downskilling.} This framework also allows for a more nuanced understanding of the \emph{de re} and \emph{de dicto} distinctions in epistemic sentences.

The structure of the paper is as follows: Section~\ref{sec:logics} is dedicated to presenting the formal syntax and semantics of our proposed logics. This section also includes a discussion on the use of \emph{epistemic de re} and \emph{de dicto} expressions within our framework. The subsequent section delves into an in-depth analysis of the computational complexity associated with the model checking problems in these logics. The paper concludes with Section~\ref{sec:conclusion}, where we offer our concluding remarks and reflections on the study.

\section{Logics}
\label{sec:logics}

We extend classical epistemic logic \cite{FHMV1995,MvdH1995} with a mechanism of epistemic skills in the models, allowing us a consistent way of modeling \emph{knowing} and \emph{forgetting}, as well as various notions of group knowledge (such as, \emph{distributed knowledge} and \emph{field knowledge}).

We fix three countably infinite sets before the introduction of formal languages. Namely, \pr for the \emph{set of atomic propositions} (\emph{atoms} for short), \ag for the \emph{set of agents} and \sk for the \emph{set of epistemic skills} (capabilities, professions, or privileges). For simplicity, these sets are unchanged throughout the paper, although it is also possible to treat them as changeable parameters of each of the languages.

\subsection{Syntax}\label{sec:syntax}

The biggest language that we introduce now, named $\lang_{CDEF+-=\equiv\boxplus\boxminus\Box}$, has its grammar given as follows:
\begin{align*}
\phi ::= &\ p \mid \neg \phi \mid (\phi \ra \phi) \mid K_a \phi \mid C_G \phi \mid D_G \phi \mid E_G \phi \mid F_G \phi \mid \\
&\ (+_S)_{a} \phi \mid (-_S)_{a} \phi \mid ({=}_S)_{a} \phi \mid ({\equiv}_b)_{a} \phi \mid
\boxplus_{a} \phi \mid \boxminus_{a} \phi \mid \Box_a \phi
\end{align*}
where $p \in \pr$, $a, b \in \ag$, $G \subseteq \ag$ is a finite nonempty group, and $S \subseteq \sk$ is a finite nonempty skill set.

As the name shows, we are interested in some of its sublanguages. The basic language \lang allows a grammar that builds recursively from atomic propositions with Boolean operators (we choose negation and implication to be primitives) and the modal operator $K_a$ (with $a \in \ag$) which is used to characterize \emph{individual knowledge}. Namely, \lang is the formal language for classical multi-agent epistemic logic.

Four types of modalities, $C_G$, $D_G$, $E_G$ and $F_G$, are introduced for \emph{common knowledge}, \emph{distributed knowledge}, \emph{mutual knowledge} and \emph{field knowledge}, respectively. In naming a language that extends the basic language, we use combinations of the letters $C$, $D$, $E$ and $F$ to indicate the inclusion common, distributed, mutual or field knowledge operators. For example, \langdf denotes the language that extends the basic language with distributed and field knowledge.

We consider four update modalities, $(+_S)_{a}$, $(-_S)_{a}$, $(=_S)_a$ and $(\equiv_b)_a$, where $a, b \in \ag$ and $S$ is a finite nonempty subset of $\sk$, which are intended to mean the action of agent $a$'s expansion with skills $S$ (\emph{upskilling}), subtraction of skills $S$ (\emph{downskilling}), assigning skill set $S$ (\emph{reskilling}) and \emph{learning} from agent $b$, respectively. These operators are self dual, as one can verify after the semantics is introduced. 

Another three operators, $\boxplus_a$, $\boxminus_a$ and $\Box_a$, are used to mean the action of $a$'s addition, subtraction and modification of an arbitrary skill set, respectively. Their dual operators are written as $\diamondplus_a$, $\diamondminus_a$ and $\Diamond_a$, respectively, but treated to be non-primitive. 

We shall use the symbols $+$, $-$, $=$, $\equiv$, $\boxplus$, $\boxminus$ and $\Box$ in subscript to signal the introduction of each of the update operators or quantifiers. For example, $\lang_{F + \boxplus}$ stands for the language that extends the basic language with field knowledge and the operators $(+_S)_{a}$ and $\boxplus_a$ (for any $a \in \ag$ and $S \subseteq \sk$).

As a result, we reach as many as $2^{11}$ ($= 2048$) languages in total, though many of the combinations may not be of our focus. Other Boolean operators are defined just as in classical logic. When we refer to a \emph{formula}, we are indicating an element of one of these languages, and its specific reference will depend on the context unless otherwise specified.

\subsection{Semantics}\label{sec:semantics}\label{sec:models}

We introduce a type of models for the interpretation of the languages.

\begin{definition}\label{def:models}
A \emph{model} is a quadruple $(W, E, C, \beta)$ where:
	\begin{itemize}
	\item $W$ is a nonempty set of (possible) worlds or states;
	\item $E : W \times W \to \wp(\sk)$, an \emph{edge function}, assigns each pair of worlds a skill set;
	\item $C: \ag \to \wp(\sk)$ is a \emph{capability function} that assigns a skill set to each agent;
	\item $\beta: W \to \wp(\pr)$ is a valuation.
	\end{itemize}
and satisfies the following two conditions:
	\begin{itemize}
		\item Positivity: for all $w, u \in W$, if $E(w, u) = \sk$, then $w = u$;
		\item Symmetry: for all $w, u \in W$, $E(w, u) = E(u, w)$.
	\end{itemize}
\end{definition}
In the above definition, the function $E$ assigns a skill set to each edge (a pair of worlds), indicating that only individuals with skills outside the set can distinguish between the pair of worlds.
The criteria for satisfaction are defined as follows.

\begin{definition}\label{def:semantics}
Given a formula $\phi$, a model $M = (W, E, C, \beta)$ and $w \in W$, we say $\phi$ is \emph{true} or \emph{satisfied} at $w$ in $M$, denoted $M, w \models \phi$, if the following hold inductively:
\begingroup
\addtolength{\jot}{-.3ex}
\begin{alignat*}{6}
&M, w \models p & \Lra &&& p \in \beta(w)\\
&M, w \models \neg \psi & \Lra &&& \text{not } M, w \models \psi \\
&M, w \models (\psi {\ra} \chi) & \Lra &&& \text{if }M, w \models \psi \text{ then } M, w \models\chi \\
&M, w \models K_a \psi & \Lra &&& \text{for all $u \in W$, if $C(a) \subseteq E(w, u)$ then $M, u \models \psi$} \\
&M, w \models E_G \psi & \Lra &&& \text{$M, w \models K_a\psi$ for all $a\in G$} \\
&M, w \models C_G \psi & \Lra &&& \text{for all positive integers $n$, $M, w \models E_G^n \psi$, with $E_G^1\psi := E_G \psi$ and $E_G^n\psi := E^1_G E_G^{n-1}\psi$} \\
&M, w \models D_G \psi & \Lra &&& \text{for all $u \in W$, if $\textstyle\bigcup_{a\in G} C(a) \subseteq E(w, u)$ then $M, u \models \psi$} \\
&M, w \models F_G \psi & \Lra &&& \text{for all $u \in W$, if $\textstyle\bigcap_{a\in G} C(a) \subseteq E(w, u)$ then $M, u \models \psi$} \\
&M, w \models (+_S)_{a} \psi & \Lra &&& (W, E, {C^{a+S}}, \beta), w \models \psi, \text{ with $C^{a+S}(a) = C(a) \cup S$ and $(\forall x \in \ag\setminus\{a\})\ C^{a+S}(x) = C(x)$}
\\
&M, w \models (-_S)_a \psi & \Lra &&& (W, E, C^{a-S}, \beta), w \models \psi, \text{ with $C^{a-S}(a) = C(a) \setminus S$ and $(\forall x \in \ag \setminus \{a\})\ C^{a-S}(x) = C(x)$}
\\
&M, w \models ({=}_S)_a \psi & \Lra &&& (W, E, C^{a = S}, \beta), w \models \psi, \text{with $C^{a=S}(a) = S$ and $(\forall x \in \ag\setminus\{a\})\ C^{a = S}(x) = C(x)$}
\\
&M, w \models ({\equiv}_b)_a \psi & \Lra &&& (W, E, C^{a \equiv b}, \beta), w \models \psi, \text{with $C^{a\equiv b}(a) = C(b)$ and $(\forall x \in \ag\setminus\{a\})\ C^{a \equiv b}(x) = C(x)$}
\\
&M, w \models \boxplus_a \psi & \Lra &&& \text{for all finite nonempty $S \subseteq \sk$, } M, w \models (+_S)_a \psi
\\
&M, w \models \boxminus_a \psi & \Lra &&& \text{for all finite nonempty $S \subseteq \sk$, } M, w \models (-_S)_a \psi
\\
&M, w \models \Box_a \psi & \Lra &&& \text{for all finite nonempty $S \subseteq \sk$, } M, w \models {({=}_{S})_a} \psi.
\end{alignat*}
\endgroup
\end{definition}

Given that $G$ is a finite group, it is clear that the formula $E_G \psi$ is logically equivalent to the $\bigwedge_{a \in G} K_a \psi$. However, this equivalence impacts both the succinctness of the language and the complexity of model checking. Consequently, $E_G \psi$ cannot be treated merely as a straightforward rewriting of $\bigwedge_{a \in G} K_a \psi$.

Note that although $({=}_S)_{a} \phi$ is not a legal formula when $S$ is the empty set $\emptyset$, we can regard it as a defined formula, i.e., $({=}_\emptyset)_{a} \phi := ({=}_S)_a (-_{S})_a \phi$ (for any qualified set $S$). In the mean time, it is not hard to verify that both $(+_\emptyset)_a \phi$ and $(-_\emptyset)_a \phi$, if allowed, are equivalent to $\phi$, so there is no need to worry about the cases with empty sets.

The logics (i.e., the sets of valid formulas) that are defined by the above semantics and correspond to our languages will bear the same names, but will be denoted using upright roman typeface, e.g., \l, $\l_{F+\boxplus}$ and $\l_{CDEF+-=\equiv\boxplus\boxminus\Box}$.

\subsection{Representation of a model and truths within it}
In this section, we describe an exemplary model and demonstrate several true formulas applicable within this model. Let $s_1, s_2, s_3, s_4, s_5 \in \sk$ represent epistemic skills, and $a, b, c \in \ag$ denote agents. The model $M = (W, E, C, \beta)$ is defined as follows:
\vspace{-1ex}
\begin{multicols}{2}
\begin{itemize}[leftmargin=1em]
\item $W = \{ w_1, w_2, w_3, w_4, w_5 \}$ is the set of possible worlds.
\item  $E : W \times W \to \wp(\sk)$ is the symmetric closure that satisfies the following:
	\begin{itemize}
	\item $E(w_1,w_1) = E(w_2,w_2) = E(w_3,w_3) = E(w_4,w_4) = E(w_5,w_5) = \{ s_1, s_2, s_3, s_4 \}$,
	\item $E(w_1,w_2) = E(w_3, w_5) = \{ s_1, s_4\}$,
	\item $E(w_1,w_3) = E(w_2, w_5) = \{ s_1, s_2, s_3\}$,
	\item $E(w_1,w_4) = \emptyset$,
	\item $E(w_1,w_5) = E(w_2, w_3) = \{ s_1 \}$,
	\item $E(w_2,w_4) = \{ s_2, s_3 \}$,
	\item $E(w_3,w_4) = \{ s_4 \}$,
	\item $E(w_4,w_5) = \{ s_2, s_3, s_4 \}$.
	\end{itemize}
\item $C$ is the capability function that assigns a skill set to each agent, $a$, $b$ and $c$:
	\begin{itemize}
	\item $C(a) = \{s_1, s_2, s_3\}$,
	\item $C(b) = \{s_2, s_3, s_4\}$,
	\item $C(c) = \{s_4\}$.
	\end{itemize}
\item $\beta$ assigns sets of propositions to each world:
	\begin{itemize}
	\item $\beta(w_1) = \{p_1, p_2\}$
	\item $\beta(w_2) = \{p_1, p_3\}$
	\item $\beta(w_3) = \{p_1, p_2, p_4\}$
	\item $\beta(w_4) = \{p_3, p_4\}$
	\item $\beta(w_5) = \{p_1,p_3,p_4\}$.
	\end{itemize}
\end{itemize}
\end{multicols}

\noindent The fact that $M$ is a model can be easily verified, and using a diagram to represent $M$ is often helpful (see Figure~\ref{fig:sim-model}). In the diagram, nodes represent worlds, and undirected edges represent accessibility relations. Each edge is labeled with the skill set that facilitate accessibility between two worlds. If an edge is labeled with an empty set, it indicates no accessibility (as between $w_1$ and $w_4$), and such edges are not drawn in the diagram. This helps in visualizing the connections and structure of the model more clearly.

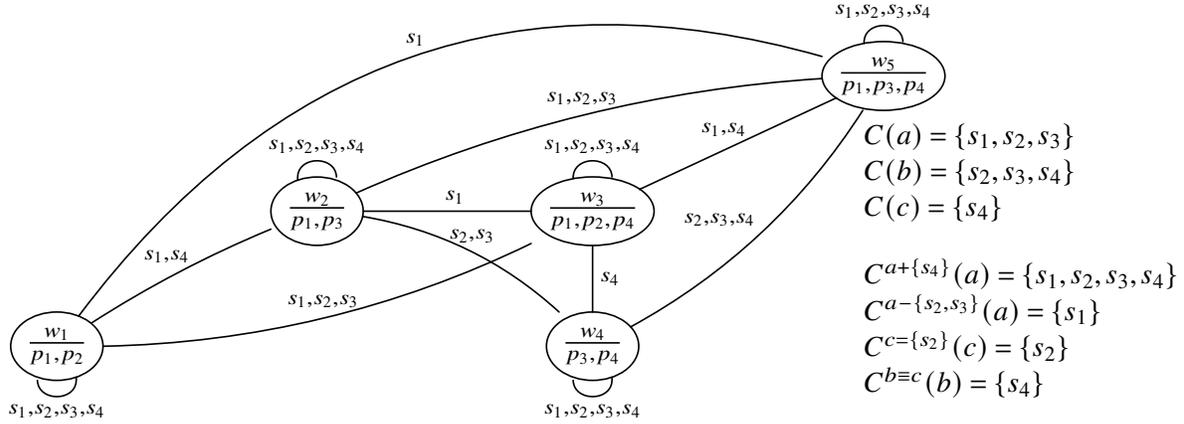
\begin{figure}[ht]
\centering
\parbox{.70\textwidth}{%
\centering
$\xymatrix@R=2.3em@C=5.8em{
&&
&*++o[F]{\frac{w_5}{p_1,p_3,p_4}}
\ar@{-}@(ul,ur)^{s_1, s_2, s_3, s_4}
\\
& *++o[F]{\frac{w_2}{p_1,p_3}}
\ar@{-}@(ul,ur)^{s_1, s_2, s_3, s_4}
\ar@{-}[r]^{s_1}
\ar@{-}@/^1pc/[dr]^{\hspace{-6pt}s_2, s_3}
\ar@{-}@/^1pc/[urr]^{s_1, s_2, s_3}
&*++o[F]{\frac{w_3}{p_1,p_2,p_4}} 
\ar@{-}@(ul,ur)^{s_1, s_2, s_3, s_4}
\ar@{-}[d]^{s_4}
\ar@{-}[ru]^{s_1,s_4}
\\
*++o[F]{\frac{w_1}{p_1,p_2}}
\ar@{-}@(dr,dl)^{s_1, s_2, s_3, s_4}
\ar@{-}@/^.3pc/[ur]^{s_1, s_4}
\ar@{-}@/^-1.2pc/[rru]^{s_1, s_2, s_3}
\ar@{-}@/^5.2pc/[uurrr]^{s_1}
&
&*++o[F]{\frac{w_4}{p_3,p_4}}
\ar@{-}@(dr,dl)^{s_1, s_2, s_3, s_4}
\ar@{-}@/^-1pc/[uur]^{s_2, s_3, s_4}
}$}
\parbox{.27\textwidth}{%
$\begin{array}{l}
\bigskip\bigskip\\
C(a) = \{s_1, s_2, s_3\}\\
C(b) = \{s_2, s_3, s_4\}\\
C(c) = \{s_4\}\bigskip\\
C^{a+\{s_4\}}(a)=\{s_1, s_2, s_3, s_4\}\\
C^{a-\{s_2, s_3\}}(a)=\{s_1\}\\
C^{c=\{s_2\}}(c)=\{s_2\}\\
C^{b\equiv c}(b)=\{s_4\}\\
\end{array}$
}
\caption{Illustration of the model $M$.}\label{fig:sim-model}
\end{figure}

Readers can verify the following logical truths in model $M = (W, E, C, \beta)$ given above:

\begin{enumerate}[itemsep=0pt]
\item $M, w_2 \models K_a p_3$, indicating that in world $w_2$, agent $a$ knows proposition $p_3$.
\item $M, w_4 \models \neg K_b p_1 \wedge \neg K_b \neg p_1$, meaning that in world $w_4$, agent $b$ does not know whether proposition $p_1$ is true or false.
\item $M, w_3 \models K_c (K_a p_3 \vee K_a \neg p_3)$, demonstrating that in world $w_3$, agent $c$ knows whether agent $a$ knows proposition $p_3$.
\item $M, w_4 \models E_{\{a,b\}} (p_3 \wedge p_4)$, showing that both agents $a$ and $b$ know propositions $p_3$ and $p_4$ in $w_4$.
\item $M, w_5 \models (\neg C_{\{a,c\}} p_1 \wedge \neg C_{\{a,c\}} \neg p_1) \wedge (\neg C_{\{a,c\}} p_2 \wedge \neg C_{\{a,c\}} \neg p_2)$, indicating that neither the truth nor the falsity of propositions $p_1$ and $p_2$ are common knowledge between $a$ and $c$ in world $w_5$.
\item $M, w_4 \models D_{\{a,b\}} (\neg p_1 \wedge p_4)$, indicating that in world $w_4$, the knowledge that proposition $p_1$ is false and $p_4$ is true is distributed between agents $a$ and $b$.
\item $M, w_4 \models \neg F_{\{a,b\}} \neg p_1 \wedge \neg F_{\{a,b\}} p_4$, showing that in world $w_4$, neither $\neg p_1$ nor $p_4$ are field knowledge for agents $a$ and $b$.
\item $M, w_5 \models \neg K_a p_4 \land (+_{\{s_4\}})_a K_a p_4$. Here, in world $w_5$, agent $a$ initially does not know $p_4$, but can learn it upon acquiring skill $s_4$.
\item $M, w_2 \models K_a p_3 \land (-_{\{s_2, s_3\}})_a \neg K_a p_3$, indicating that in world $w_2$, agent $a$ knows $p_3$ but would forget it if she loses skills $s_2$ and $s_3$.
\item $M, w_1 \models E_{\{a, b\}} (\neg K_c p_2 \land (=_{\{s_2\}})_c K_c p_2))$. This means that in world $w_1$, it is mutual knowledge between agents $a$ and $b$ that, $c$ does not know $p_2$, but would know it if her skill set is $\{s_2\}$.
\item $M, w_1 \models (\equiv_c)_b \bigwedge_{p\in \{p_1,\dots,p_4\}}(F_{\{b,c\}} p \lra K_b p)$. This result suggests that in world $w_1$, if agent $b$ changes her skill set to match that of agent $c$, her knowledge will align with the field knowledge shared between them.
\item $M, w_5 \models \diamondplus_a K_a p_4$, indicating that in world $w_5$, there exists a potential skill upgrade under which agent $a$ can come to know $p_4$.
\item $M, w_3 \models \diamondminus_b \bigwedge_{p\in \{p_1,\dots,p_4\}} (\neg C_{\{a, b\}} p \land \neg C_{\{a, b\}} \neg p)$, highlighting that in world $w_3$, it is possible through some downskilling for agent $b$ to reach a state where no propositions $p_1$ through $p_4$ are common knowledge between agents $a$ and $b$.
\item $M, w_2 \models K_c p_1\land \neg K_c p_3\land \Diamond_c (\neg K_c p_1\land K_c p_3)$, indicating that in world $w_2$, it is currently the case that agent $c$ knows $p_1$ and does not know $p_3$, but there is a possible skill update which would make agent $c$ unaware of $p_1$ while becoming aware of $p_3$.
\end{enumerate}

\subsection{Epistemic de re and de dicto}

The distinction between epistemic \emph{de re} and \emph{de dicto} modalities was already discussed in \cite{Wright1951}, with \emph{de re} modalities concerning whether a specific thing possesses or lacks a certain property, and \emph{de dicto} modalities concerning whether a proposition is true or false. Subsequently, as \cite{Quine1956} suggests, this distinction becomes more apparent in a formal language when quantifiers over terms are utilized. In the realm of epistemic logic, a \emph{de re} sentence can be expressed as, ``There exists a term $x$ such that an entity knows that $x$ possesses or lacks a certain property.'' Conversely, a \emph{de dicto} sentence can be formulated as, ``An entity knows that there exists a term which possesses or lacks a certain property.''

In the logics introduced in this paper, we are not only able to distinguish between \emph{de re} and \emph{de dicto} modalities, but can also delineate two specific types of de re sentences (compare with the case in Group Announcement Logic \cite[Section~6]{ABDS2010}):
\begin{itemize}
\item Knowing de dicto: ``Agent $a$ knows (with her current skills) that there exists a set $S$ of skills such that, with $S$, she can achieve $\phi$ in world $w$ of model $(W, E, C, \beta)$.''\\
Formally, this is expressed as: $(\forall u \in W) [C(a)\subseteq E(w,u) \Ra (\exists S \subseteq \sk)\ (W,E,C^{a+S},\beta),u\models \phi]$.
\item Explicitly knowing de re: ``There exists a set $S$ of skills such that agent $a$ knows with her current skill set, that with $S$ in addition, she can achieve $\phi$ in world $w$ of model $(W, E, C, \beta)$.''\\
Formally, this is represented as: $(\exists S\subseteq\sk)(\forall u\in W)[C(a)\subseteq E(w,u) \Ra (W,E,C^{a+S},\beta),u\models \phi]$.
\item Implicitly knowing de re: ``There exists a set $S$ of skills such that agent $a$ knows, with the addition of $S$ to her skill set, that she can achieve $\phi$ in world $w$ of model $(W, E, C, \beta)$.''\\
Formally, this is depicted as: $(\exists S\subseteq\sk)(\forall u\in W)[C^{a+S}(a)\subseteq E(w,u) \Ra (W,E,C^{a+S},\beta),u\models \phi]$.
\end{itemize}
Although the distinction between \emph{de dicto} and \emph{de re} knowledge remains clear, the nuanced difference between \emph{implicit} and \emph{explicit} \emph{de re} knowledge hinges on whether the skills from the skill set $S$ are included for the agent to formulate her knowledge.

These distinctions elucidate the complex interplay between knowledge and capabilities in dynamic epistemic scenarios, highlighting subtle differences in how agents process information depending on their skill sets and the nature of their knowledge. All three types of knowledge are expressible using the formal languages we have introduced. Here is how each type is represented:

\begin{proposition}\label{prop:de-re-de-dicto}\ 
\begin{enumerate}
\item\label{it:de-dicto} Knowledge de dicto is expressed by the formula $K_a\diamondplus_a\phi$.
\item\label{it:de-re} Explicit knowledge de re is expressed by the formula $(\equiv_a)_c\diamondplus_c K_a (\equiv_c)_a \phi$ (where $c$ is not in $\phi$).
\item\label{it:knowable} Implicit knowledge de re is expressed by the formula $\diamondplus_a K_a\phi$.
\end{enumerate}
\end{proposition}
\noindent\textit{Proof.}
Clauses \ref{it:de-dicto} and \ref{it:knowable} are straightforwardly validated by the semantics. We focus here on clause \ref{it:de-re}, where $c$ is any agent not appearing in $\phi$:

$\begin{array}{cl}
& (\exists S\subseteq\sk)(\forall u\in W)\ C(a)\subseteq E(w,u) \Ra (W,E,C^{a+S},\beta),u\models \phi \\
 \Longleftrightarrow & (\exists S\subseteq\sk)(\forall u\in W)\ C(a)\subseteq E(w,u) \Ra (W,E,((C^{c\equiv a})^{c+S})^{a\equiv c},\beta),u\models \phi \\
 \Longleftrightarrow & (\exists S\subseteq\sk)(\forall u\in W)\ C(a)\subseteq E(w,u) \Ra (W,E,(C^{c\equiv a})^{c+S}, \beta), u\models (\equiv_c)_a\phi \\
 \Longleftrightarrow & (\exists S\subseteq\sk)(W,E,(C^{c\equiv a})^{c+S},\beta), w \models K_a (\equiv_c)_a \phi \\
 \Longleftrightarrow & (W,E,C^{c\equiv a},\beta),w\models \diamondplus_c K_a(\equiv_c)_a\phi \\
 \Longleftrightarrow & (W,E,C,\beta),w\models (\equiv_a)_c\diamondplus_c K_a(\equiv_c)_a\phi. \\
\end{array}$

\paragraph{Examples of different types of knowledge expressed by formulas}

The formula $D_G \diamondplus_a \diamondminus_b \phi$ says that, ``It is group $G$'s distributed knowledge that, with the addition of certain skills by agent $a$, it becomes achievable that, even with the loss of certain skills by agent $b$, $\phi$ can still be achieved.'' This use pertains to be \emph{de dicto}.
The formula $(\equiv_a)_c\Diamond_c K_a (\equiv_c)_a \phi$ (where $c$ does not appear in $\phi$) expresses that, ``There exists a skill set such that agent $a$ knows that, with exactly this skill set, $a$ can achieve $\phi$.'' This use is \emph{explicitly de re}.
The formula $\Diamond_a K_a\phi$ says that, ``There is an update of agent $a$'s skill set through which $a$ knows she can make $\phi$ true.'' This is \emph{implicitly de re}.

\begin{remark}
For simplicity, the initial definitions of \emph{knowledge de dicto}, \emph{explicit} and \emph{implict knowledge de re} have been presented primarily for individual knowledge using the operator $K_a$ and the actions of knowing represented by the quantifier $\boxplus_a$. These concepts can be readily extended to include:
\begin{itemize}
\item \emph{Group knowledge}, utilizing operators such as $C_G$, $D_G$, $E_G$ and $F_G$,
\item \emph{Quantifiers over actions of downskilling} and \emph{reskilling}, represented by $\boxminus_a$ and $\Box_a$ respectively,
\item \emph{Nested actions} and \emph{dynamic changes} among agents.
\end{itemize}
For example,
the formula $F_{G}\diamondminus_{a_1}\diamondplus_{a_2}\Diamond_{a_3}\diamondminus_{a_4}\phi$ represents an \emph{epistemic de dicto} statement involving field knowledge and multiple actions (upskilling, downskilling and reskilling) among different agents. 
The expression $(\equiv_{d_1})_{c_1}\diamondplus_{c_1}(\equiv_{d_2})_{c_2}\diamondplus_{c_2}(\equiv_{d_3})_{c_3}\diamondplus_{c_3} E_{I}(\equiv_{c_1})_{d_1}(\equiv_{c_2})_{d_2}(\equiv_{c_3})_{d_3}\phi$ captures \emph{explicit knowledge de re} involving nested contexts and multiple agents, linked to mutual knowledge.
Similarly, the formula $\diamondplus_{b_1}\Diamond_{b_2}\diamondminus_{b_3}D_{H}\phi$ illustrates \emph{implicit knowledge de re} involving a sequence of updates and distributed knowledge.
In these examples, the agents $a_1, a_2, a_3, a_4, b_1, b_2, b_3, c_1, c_2, c_3, d_1, d_2, d_3$ are not specifically restricted to being within or outside the groups $G$, $H$ or $I$. This flexibility allows for a broad application of the concepts across various contexts and group dynamics.
\end{remark}

In dynamic epistemic logic, the distinction between \emph{knowing de dicto} and \emph{knowing de re} is enriched through the use of quantifiers for updates, closely aligning with the philosophical inquiries into \emph{knowing that} versus \emph{knowing how}. While previous solutions such as those presented in \cite{BBDHHL2008,ABDS2010,BDK2013} primarily adopt a syntactical approach, our logic introduces a semantical perspective, providing an alternative to the topological semantics discussed by \cite{WA2013SSPAL,BOS2017}.

\section{Complexity of Model Checking}

In this section, we study the computational complexity of the model checking problem for the logics introduced in the preceding sections. The \emph{model checking problem} for a logic involves verifying whether a specified formula $\phi$, within a given finite model $M$ and at a particular world $w$ in the model, holds true; formally, whether $M, w \models \phi$.

\subsection{The input}

We define the measure of the input. The \emph{length} of a formula $\phi$, denoted $|\phi|$, is defined to be the number of symbols that occur in $\phi$ (including the symbols for brackets), just as in \cite[Section 3.1]{FHMV1995}; or more precisely defined inductively by the structure of $\phi$, i.e., when $\phi$ is:
\begin{itemize}
	\item An atomic formula $p$: $|p| = 1$;
	\item Negation $\neg\psi$: $| \neg \psi | = |\psi| +1$;
	\item Implication $(\psi \ra \chi)$: $| (\psi \ra \chi) | = |\psi| + |\chi| + 3$;
	\item Individual knowledge $K_a \psi$: $|K_a\psi| = |\psi| + 2$;
	\item Group knowledge: $|C_G\psi| = |\psi| + 2|G| + 2$,  and similarly for $D_G \psi$, $E_G \psi$ and $F_G\psi$; e.g., $| (p \ra C_{\{a,b,c\}} q) | = 13$;
	\item An update modality: $|(+_S)_a \psi| = 2|S| + |\psi| + 5$, similarly for $(-_S)_a \psi$ and $(=_S)_a \psi$, and $|(\equiv_b)_a \psi| = |\psi| + 5$;
	\item A quantifier: $|\boxplus_a \psi| = |\psi| + 2$, and also for $\boxminus_a \psi$ and $\Box_a \psi$.
\end{itemize}

The \emph{size} of a model $M = (W,E,C,\beta)$, denoted $|M|$, is defined as the sum of the following components:
\begin{itemize}
\item $|W|$: the size of the domain;
\footnote{Model checking is typically impractical for infinite sets due to computational limitations; therefore, we restrict our analysis to finite sets. This is consistently applied in the following discussions as well.}
\item $|E|$: since $E$ consists of triples $(w, u, S)$ where $w, u \in W$ and $S \subseteq \sk$, the size of $E$ is determined by the number of the symbols used to denote this set;
\item $|C|$ with respect to a given set $A$ of agents: $C$ is composed pairs $(a, S)$ where $a \in A$ and $S \subseteq \sk$; the size of $C$ is the count of all symbols used for its representation;
\footnote{Theoretically, the function $C$ maps each agent (from an infinite set) to a specific skill set. This mapping is not feasible with finite input, but in practical scenarios, we limit the number of agents. It is essential to ensure that the set $A$ includes all agents relevant to the formula being checked.}
\item $|\beta|$: the function $\beta$ consists of pairs $(w, \Phi)$ where $w \in W$ and $\Phi \subseteq \pr$; the size of $\beta$ is the number of the symbols used to represent this set.
\end{itemize}

Finally, for formula $\phi$ and model $M$ (with a designated world $w$), the \emph{size of the input} is $|\phi| + |M| + 1$.

\subsection{Model checking for logics without quantifiers: in P}
\label{sec:mc}

We commence by presenting a polynomial-time algorithm designed to ascertain the truth of classical epistemic formulas in a specified world within a given model, addressing the model checking problem for \l. Subsequently, we enhance the algorithm to incorporate group knowledge modalities. This extension allows us to establish that the model checking problem for $\l_{CDEF}$ fall within the complexity class P. We then proceed to further broaden our results to encompass update modalities, achieving the results for the model checking problems for $\l_{CDEF+-=\equiv}$ and all of its sublolgics.

\subsubsection{Model checking in \l}

Given a model $M = (W, E, C, \beta)$, a world $w \in W$ and a formula $\phi$, we decide whether $M, w \models \phi$. In order to do so, we present an algorithm (Algorithm~\ref{alg:val}) for calculating $Val(M,\phi)$, the \emph{truth set} of $\phi$ in $M$, i.e., $\{ x \in W \mid M, x \models \phi\}$. The question about whether $M, w \models \phi$ holds is thus reduced to the membership testing in $Val (M, \phi)$, which takes at most $|W|$ steps in addition to the time costs on computing $Val (M, \phi)$.%
	
\begin{algorithm}
\caption{Function $Val (M, \phi)$: computing the truth set for basic formulas}\label{alg:val}
\small
\vspace{-1em}
\begin{multicols}{2}
\begin{algorithmic}[1]
\Require model $M = (W, E, C, \beta)$ and formula $\phi$
\Ensure $\{ x \mid M, x \models \phi \}$
\Initialize{$tmpVal \gets \emptyset$}

\If{$\phi=p$} \Return $\{ x \in W \mid p \in \beta(x)\}$

\ElsIf{$\phi=\neg\psi$} \Return $W \setminus Val(M,\psi)$

\ElsIf{$\phi=\psi\to\chi$}
	\State\Return $(W \setminus Val(M,\psi)) \cup Val(M,\chi)$

\ElsIf{$\phi=K_a\psi$}
\ForAll{$x\in W$}
\Initialize{$n \gets \True$}
\ForAll{$y \in W$}
\If{{\footnotesize $C(a) \subseteq E (x, y)$} \algand {\footnotesize~$y \notin Val (M, \psi)$}} $~~~~n \gets \False$
\EndIf
\EndFor
\If{$n = \True$} $tmpVal \gets tmpVal \cup \{x\}$
\EndIf
\EndFor
\State \Return $tmpVal$ \Comment{This returns $\{x\in W \mid \forall y \in W: C(a)\subseteq E (x, y) \Rightarrow y \in Val (M,\psi)\}$}
\EndIf
\end{algorithmic}
\end{multicols}
\vspace{-0.8em}
\end{algorithm}
	
	It is not hard to verify that $Val(M,\phi)$ is indeed the set of worlds of $M$ at which $\phi$ is true. In particular, in the case for the $K_a$ operator,
	\[
	\begin{array}{llll}
		M, w \models K_a \psi & \iff & \forall y \in W: C(a) \subseteq E (w, y) \Rightarrow M, y \models \psi \\
		& \iff & \forall y \in W: C(a) \subseteq E (w, y) \Rightarrow y \in Val (M, \psi)&\text{(IH)}\\
		& \iff & w \in \{ x \in W \mid \forall y \in W: C(a) \subseteq E (x, y) \Rightarrow y \in Val(M,\psi)\} 
	\end{array}
	\]

The cost for computing $Val(M,\phi)$ is in polynomial time. In the case for $K_a \psi$---the most time-consuming case here---there are two while-loops over $W$, and checking $C(a) \subseteq E (x, y)$ costs at most $|C| \cdot |E|$ steps, and the membership checking $y \notin Val(M,\psi)$ (when $Val(M,\psi)$ is at hand) takes at most $|W|$ steps; so this case costs at most $|W|^2 \cdot (|C| \cdot |E| + |W|)$. Moreover, the algorithm for computing $Val(M,\phi)$ calls itself recursively, but only for a subformula of $\phi$, and the maximum number of recursion is bounded by $|\phi|$, i.e., the length of $\phi$. So the total time cost for computing $Val(M,\phi)$ is $|W|^2 \cdot (|C| \cdot |E| + |W|) \cdot |\phi|$. Considering the input size, we find that the total time cost is within $O(n^5)$. So the following lemma holds. 

\begin{lemma}\label{lem:mc-el}
The model checking problem for \l is in P.
\end{lemma}

\subsubsection{Model checking group knowledge}
Building on the previous result, we now aim to encompass scenarios that include group knowledge. To facilitate this extension, we will first introduce a definition and a couple of lemmas that underpin it.

\begin{definition}\label{def:trans-e}
	For a formula $\phi$, let $\ab_\phi = \{G \mid \text{``$E_G$'' or ``$C_G$'' appears in $\phi$}\}$. For a model $M = (W, E, C, \beta)$,
	\begin{itemize}
	\item For all worlds $w, u \in W$, $E_\phi (w, u) = E (w, u) \cup \{G \in A_\phi \mid (\exists a\in G)\ C(a) \subseteq E(w,u) \}$,
	\item For all worlds $w, u \in W$, $E_{\phi}^+(w,u) = E_\phi(w,u) \cup \{G\in A_\phi \mid (\exists n \geq 1)(\exists w_0,\dots,w_n\in W)\ w_0=w \text{ and } w_n=u \text{ and } G\in \bigcap_{0\leq i< n}E_\phi(w_i,w_{i+1})\}$,
	\end{itemize}
	where without loss of generality we assume that $A_\phi \cap \ag = \emptyset$. For short, we write $M_\phi^+$ for $(W, E^+_\phi, C, \beta)$.
\end{definition}
It should be noted that the above definition involves an abuse of notation by treating groups of agents as skills. To ensure formal correctness, a one-to-one mapping can be defined from each group to a new skill in \sk.

\begin{proposition}
For any model $M$ and any formula $\phi$, $M^+_\phi$ is a model.\qed
\end{proposition}

\begin{lemma}\label{lem:trans-e}
Given formulas $\phi$ and $\chi$, a group $G$, a model $M$ and a world $w$ of $M$:
\begin{enumerate}
\item\label{it:trans-e1} $M,w\models\phi$ iff $M^+_{\chi},w\models\phi$;
\item\label{it:trans-e3} If ``$C_G$'' appears in $\chi$, then $M,w\models C_G\phi$ iff $M,u\models \phi$ for any world $u$ such that $G\in E^+_{\chi}(w,u)$.
\end{enumerate} 
\end{lemma}

\begin{proof}
\ref{it:trans-e1}. For any agent $a$, formula $\chi$ and world $w, u$, we have $C(a)\subseteq E(w, u)$ iff $C(a)\subseteq E_\chi(w, u)$ iff $C(a) \subseteq E^+_{\chi}(w, u)$. Thus it is easy to verify that $(M, w)$ and $(M^+_{\chi},w)$ satisfy exactly the same formulas.
	
\ref{it:trans-e3}. We first verify the base case for $C_G\phi$:
$$\begin{array}[t]{rcll}
M, w \models E_G \phi & \iff & \text{for any $a \in G$, $M, w\models K_a\phi$} & \\
& \iff & \text{for any $a\in G$ and $u\in W$, $C(a)\subseteq E(w, u)$ implies $M,u \models \phi$} & \\
& \iff & \text{for any $u\in W$ and $a\in G$, $C(a)\subseteq E(w, u)$ implies $M,u \models \phi$} & \\
& \iff & \text{for any $u\in W$, $M,u \models \phi$ if $C(a)\subseteq E(w,u)$ for some $a\in G$} & \\
& \iff & \text{for any $u\in W$, $G\in E_\chi(w,u)$ implies $M,u \models \phi$} & \\
& \iff & \text{$M, u\models \phi$ for any world $u$ such that $G\in E_\chi(w, u)$} &\\[1ex]

\text{and so}\quad	M, w \models C_G \phi & \iff & \text{$M, w\models E^k_G\phi$ for all $k\in \mbN^+$} & \\
& \iff & \text{$M, u\models \phi$ for any world $u$ such that $G\in E^+_{\chi}(w, u)$} & (*)
\end{array}$$
where $(*)$ can be shown as follows: Suppose $M, w \not \models E^n_G \phi$ for some $n\in\mbN^+$, then by induction on $n$, we have $w_1, \dots, w_n \in W$ such that $M, w_n \not\models \phi$ and $G \in E_\chi(w, w_1)\cap\bigcap_{1\leq i< n}E_\chi(w_i, w_{i+1})$. Hence $M, w_n\not\models\phi$ and $G\in E^+_\chi(w, w_n)$. Suppose $M, u\not\models \phi$ for a world $u$ such that $G\in E^+_\chi(w, u)$, w.l.o.g, assume that there exist $w_0,\dots, w_n\in W$ such that $w_0=w$, $w_n=u$, $G\in \bigcap_{0\leq i< n} E_\chi(w_i, w_{i+1})$ and $M, w_n \not\models \phi$. Thus using the above result $n$ times we have $M, w\not\models E^n_G \phi$.
\end{proof}

\begin{lemma}\label{thm:complexity}
The model checking problem for $\l_{CDEF}$ (hence for all of its sublogics) is in P.
\end{lemma}
\begin{proof}
It suffices to provide a polynomial algorithm for the types of formulas $C_G \psi$, $D_G \psi$, $E_G \psi$ and $F_G \psi$. The details are given in Algorithm~\ref{alg:val-dcx}.
\begin{algorithm}
\caption{Function $Val(M,\phi)$ extended: cases with group knowledge operators}\label{alg:val-dcx}
\vspace{-1em}
\begin{multicols}{2}
\small
\begin{algorithmic}[1]
	\Initialize{$temVal \gets \emptyset$}
	\If{...} ... \Comment{Same as in Algorithm~\ref{alg:val}}
	\ElsIf{$\phi = C_G\psi$}
	\ForAll{$x \in W$}
	\Initialize{$n \gets \True$}
	\ForAll{$y \in W$}
	\If{$G\in E^+_{\phi} (x, y)$ \algand $y \notin Val (M, \psi)$}
		\State $n \gets \False$
	\EndIf
	\EndFor
	\If{$n = \True$}
		\State $tmpVal \gets tmpVal \cup \{x\}$
	\EndIf
	\EndFor
	\State \Return $tmpVal$ \quad \Comment{Returns $\{x \in W\mid \forall y \in W: G \in E^+_{\phi}(x, y) \Rightarrow y \in Val(M,\psi)\}$}
	\ElsIf{$\phi=D_G\psi$}
	\ForAll{$x \in W$}
	\Initialize{$n \gets \True$}
	\ForAll{$y \in W$}
	\If{$\bigcup_{a\in G} C(a) \subseteq E(x,y)$ \algand \\ $~~~y \notin Val(M,\psi)$}
		\State {$n \gets \False$}
	\EndIf
	\EndFor
	\If{$n = \True$}
		\State $tmpVal \gets tmpVal \cup \{x\}$
	\EndIf
	\EndFor
	\State \Return $tmpVal$ \quad \Comment{Returns $\{x \in W\mid \forall y \in W: \bigcup_{a\in G} C(a) \subseteq E (x, y) \Rightarrow y \in Val(M,\psi)\}$}

	\ElsIf{$\phi=E_G\psi$}
	\ForAll{$x \in W$}
	\State {{\bf initialize} $n \gets \True$}
	\ForAll{$y \in W$}
	\If{$G\in E_\phi (x, y)$ \algand $y \notin Val(M,\psi)$}
		\State {$n \gets \False$}
	\EndIf
	\EndFor
	\If{$n = \True$} {$tmpVal \gets tmpVal \cup \{ x \}$}
	\EndIf
	\EndFor
	\State \Return $tmpVal$ \quad \Comment{Returns $\{t\in W\mid \forall u \in W: G\in E_\phi(t,u)\Rightarrow u\in Val(M,\psi)\}$}

	\ElsIf{$\phi = F_G\psi$}
	\ForAll{$x \in W$}
	\Initialize{$n \gets \True$}
	\ForAll{$y \in W$}
	\If{$\bigcap_{a\in G}C(a)\subseteq E (x, y)$ \algand \\ $~~~y \notin Val(M,\psi)$}
		\State {$n \gets \False$}
	\EndIf
	\EndFor
	\If{$n = \True$} {$tmpVal \gets tmpVal \cup \{ x \}$}
	\EndIf
	\EndFor
	\State \Return $tmpVal$ \quad \Comment{Returns $\{x \in W\mid \forall y \in W: \bigcap_{a\in G} C(a) \subseteq E (x, y) \Rightarrow y \in Val(M,\psi)\}$}
	\EndIf
\end{algorithmic}
\end{multicols}
\vspace{-1em}
\end{algorithm}
As in the proof of Lemma \ref{lem:mc-el}, checking $C(a) \subseteq E(t,u)$ costs at most $|C| \cdot |E|$ steps, here we furthermore need to calculate the cost caused by group knowledge operators.
	
For $D_G$ and $F_G$, notice that the number of agents in any group $G$ that appears in $\phi$ is less than $|\phi|$, so checking $\bigcup_{a \in G} C(a)\subseteq E(t,u)$ and $\bigcap_{a \in G} C(a)\subseteq E(t,u)$ costs at most $|C| \cdot |E| \cdot |\phi|$ steps. Thus for the logics extended with these operators, the complexity for model checking would not go beyond P.
	
For $E_G$ and $C_G$, we need to ensure that there is a polynomial-time algorithm for computing $E_\phi(w, u)$ and $E^+_{\phi}(w, u)$ and checking whether $G$ is an element of them. By Definition~\ref{def:trans-e} and Lemma~\ref{lem:trans-e}, computing the set $A_\phi$ costs at most $|\phi|$ steps, since there are at most $|\phi|$ modalities appearing in $\phi$; moreover, the size of $G$ is at most $|\phi|$. To compute $E_\phi(w, u)$ for any given $w$ and $u$, it costs at most $|E|$ steps to compute $E(w, u)$ and at most $|\phi| ^ 2 \cdot |C| \cdot |E|$ steps to check for every $G \in A_\phi$ whether there exists $a \in G$ such that $C(a) \subseteq E(w, u)$. So the cost of computing the whole function $E_\phi$ can be finished in at most $|W|^2 \cdot (|E| + |\phi| ^ 2 \cdot |C| \cdot |E|)$ steps.
Now we consider the computation of $E^+_{\phi}$. Assume that we have a string that describes $E_\phi$, then we check for all pairs $(x, y), (y, z)\in W^2$ whether there exists a ``$G$'' appearing in $\phi$ such that $G\in E_\phi (x, y) \cap E_\phi (y, z)$; if it is, we add $G$ as a member of $E_\phi (x, z)$. Keep doing this until $E_\phi$ does not change any more. Every round of checking takes at most $2|\phi|^2 \cdot |W|^3$ steps, and it will be stable in at most $|\phi|\cdot|W|^2$ rounds. Then we obtain the function $E^+_{\phi}$ as we want. Every membership checking for $G \in E^+_{\phi}(w, v)$ is finished in polynomial steps. So the whole process is still in P.
\end{proof}

\subsubsection{Model checking formulas with update modalities}

As we address the case involving update modalities, let us consider a model $M = (W, E, C, \beta)$ and a world $w \in W$, and examine the formulas $(+_S)_a \psi$, $(-_S)_a \psi$, $(=_S)_a \psi$ and $(\equiv_b)_a \psi$. According to the semantics provided in Definition~\ref{def:semantics},
$$M, w \models (+_S)_a \psi \iff M^{a+S}, w \models \psi$$
where $M^{a+S} = (W, E, C^{a+S}, \beta)$ is defined such that 
\[
C^{a+S}(x) =
\begin{cases}
C(x), & \text{if } x \neq a \\
C(a) \cup S, & \text{if } x = a
\end{cases}
\]
From this, we deduce that verifying whether \( M, w \models (+_S)_a \psi \) is reducible to checking if \( M^{a+S}, w \models \psi \), effectively eliminating the leftmost update modality from consideration. An algorithm that invokes the model checking algorithm on the latter can be executed in linear time since it involves generating the updated model \( M^{a+S} \) directly from the original model \( M \) and considering the new formula \( \psi \), which is a substring of the original formula. Hence, the complete algorithm, including the invocation of the model checking, will conclude within polynomial time.

The cases with \( (-_S)_a \psi \), \( (=_S)_a \psi \) and $(\equiv_b)_a \psi$ follow a similar process, with the distinction that each involves a different modification to the model. Nonetheless, the computational cost remains within P for both scenarios. This leads us to the following theorem:

\begin{theorem}
The model checking problems for $\l_{CDEF+-=\equiv}$ and all of its sublogics are in P.
\end{theorem}

\subsection{Model checking quantified formulas: PSPACE complete}

We demonstrate the PSPACE hardness by reducing, in polynomial time, the problem of undirected edge geography (UEG) -- a variant of the generalized geography \cite{Schaefer1978,LS1980} -- to the model checking problem for any of $\l_{\boxplus}$, $\l_{\boxminus}$ or $\l_{\Box}$, since UEG is a game for which determining a winning strategy is known to be PSPACE complete \cite{FSU1993}. The PSPACE upper bound is established using a polynomial space algorithm that builds upon the algorithms introduced earlier.

Let $G=(D, R)$ be an undirected graph; i.e., $D$ is a finite nonempty set, and $R$ is a symmetric and irreflexive relation on $D$.
Given a node $d \in D$, the pair $(G, d)$ is referred to as a \emph{rooted undirected graph}.
The undirected edge geography (UEG) game on $(G, d)$ involves two players, and unfolds as follows.
\begin{enumerate}
\item Player I's Move: Player I starts by selecting edge $\{ d, d_1 \} \in R$. If no such edge exists, the game ends and Player II wins as Player I cannot make a valid move.
\item Player II's Move: After Player I''s move selecting an edge $\{ d_i, d_{i+1} \}$, Player II must choose an edge $\{ d_{i+1}, d_{i+2} \}$ that has not been chosen in previous moves. If Player II cannot make such a move, the game ends and Player I wins.
\item Alternating Turns: After Player II's move selecting an edge $\{ d_j, d_{j+1} \}$, it is Player I's turn again to choose an edge $\{ d_{j+1}, d_{j+2} \}$ not previously chosen. If Player I cannot make such a move, the game ends and Player II wins.
\item Repeat Step 2: The game continues by alternating turns following the process described in step 2.
\end{enumerate}
Alternatively, UEG game on $(G, d)$  can be recursively defined by modifying the graph after each move:
\begin{itemize}
\item The current player chooses an edge $\{ d, d' \} \in R$; if this is impossible, he loses, and the game ends.
\item The game then proceeds with the opposing player starting a new game on $(G', d')$ where $G' = (D, R \setminus \{ \{ d, d' \} \} )$.
\end{itemize}
The \emph{UEG problem}, based on a rooted undirected graph, aims to determine whether Player I possesses a winning strategy.

\begin{definition}[induced model]
Let  $G = (D, R)$ represent an undirected graph. For each edge $\{ x, y\} \in R$, assign a unique epistemic skill $s_{\{x, y\}} \in \sk$ (ensuring that $s_{\{x', y'\}} \neq s_{\{x'', y''\}}$ for any distinct unordered pairs $\{ x', y' \}$ and $\{ x'', y'' \}$), and for each node $x \in D$, assign a unique atomic proposition $p_x \in \pr$ (ensuring that $p_{x'} \neq p_{x''}$ for any distinct nodes $x'$ and $x''$). 

Define the \emph{induced model} $M_{G}$ as the tuple $(D, E, C, \beta)$ where:
\begin{itemize}
\item $E$: For every $x, y \in D$, if $\{x, y\} \in R$, then $E(x, y) = \{s_{\{x, y\}} \}$; otherwise, $E (x, y) = \emptyset$;
\item $C$: For all agents $a$, $C(a) = \emptyset$;
\item $\beta$: For each node $x \in D$, $\beta(x) = \{p_x\}$.
\end{itemize}
\end{definition}

This model $M_{G}$ is well-defined and compactly represents the relationships and properties within the graph $G$. The size of $E$ is $O(|D|^2)$ due to the pairwise relationship between nodes, while the size of $\beta$ is $O(|D|)$, reflecting the unique property assignment per node. The size of $C$ remains $O(|D|)$, given that only a limited number of agents are actually be utilized, as confirmed by the following definition and the definition of the size of the input.

\begin{definition}[induced formula]\label{def:indf}
Let  $G = (D, R)$ be an undirected graph. Consider $n$ agents $a_1, \dots, a_n \in \ag$, where $n$ is the smallest positive even number greater than or equal to $|R|$. For each $i$ where $1 \leq i \leq n$ (for $\phi_i$, only consider even numbers $i$), define:
$$\begin{array}{lll}
\psi_{i} & := & \neg K_{a_i} \bot \wedge \bigvee_{x\in D} K_{a_i} p_x\\
\chi_{i} & := & \bigvee_{x, y\in D \text{ with } x \neq y, 1 \leq j < i} (p_x \wedge \hat K_{a_j} p_y \wedge K_{a_i} p_y)\\[.5ex]
\phi_{i} & :=  & \begin{array}[t]{l}
\diamondplus_{a_1} ( \psi_{1} \wedge \neg \chi_{1} \wedge K_{a_1}\boxplus_{a_2}( \neg \psi_{2} \vee \chi_{2} \vee 
\\
\qquad\hat{K}_{a_2}\diamondplus_{a_3}(\psi_{3} \wedge \neg\chi_{3} \wedge K_{a_3}\boxplus_{a_4}(\neg \psi_{4} \vee \chi_{4} \vee
\\
\qquad\qquad\hat{K}_{a_4}\diamondplus_{a_5}(\psi_{5} \wedge \neg \chi_{5} \wedge K_{a_5}\boxplus_{a_6} (\neg \psi_{6} \vee \chi_{6} \vee
\\
\qquad\qquad\qquad\dots
\\
\qquad\qquad\qquad\qquad\hat{K}_{a_{i-2}}\diamondplus_{a_{i-1}}(\psi_{i-1} \wedge \neg \chi_{i-1} \wedge K_{a_{i-1}}\boxplus_{a_i}(\neg\psi_{i} \vee \chi_{i}))\cdots)))))).
\end{array}
\end{array}$$
In the above, $\hat K_a$ is the dual of $K_a$. The \emph{induced formula} $\phi_G$ for the graph $G$ is defined as $\phi_n$.
\end{definition}

Let us try to understand the induced formula. In a game, each agent $a_i$ plays in the $i$-th move. The formulas $\psi_{i}$ represents the condition where the player $a_i$ at the $i$-th move chooses exactly one edge from the current node. The formula $\chi_{i}$ captures the scenario where the edge chosen by player $a_i$ at the $i$-th move has been selected in a previous move, thus representing an invalid game move under the new edge rule. The conjunction $\psi_i \wedge \neg \chi_i$ ensures that each move in the game involves selecting a new, unvisited edge. As for complexity, the length of $\psi_{i}$ is in $O(|D|)$, as it involves a disjunction over each node in $D$. The length of $\chi_{ i}$ is in $O(|D|^2 \cdot |R|)$. The overall formula $\phi_G$ thus has its length in $O(|D|^2 \cdot |R|^2)$.

The formula $\phi_{G}$ constructs a logical framework that mirrors the gameplay in an undirected graph:
\begin{itemize}
\item $\diamondplus_{a_1}$: Indicates the potential for player $a_1$ to make a valid move by upskilling.
\item $\psi_1 \wedge \neg \chi_1$: Ensures that $a_1$'s choice is a new edge (valid move).
\item $K_{a_1} \boxplus_{a_2}$: Player $a_1$ ensures that no matter how player $a_2$ responds (upskills), the game's next state must be described by the formula that follows. And that formula describes that either $a_2$ does not find a new edge to choose (leading to the end the game), or, if $a_2$ chooses a new edge, then the formula starting with $\hat K_{a_3} \diamondplus_{a_3}$ must hold, indicating a situation similar to the first clause above (but for $a_3$).
\end{itemize}
This recursive and intertwined structure of $\phi_G$ effectively captures the strategic progression of the game, with each player's move affecting the possible moves of the next player, all within the framework of an undirected graph where each node represents a game state or choice.

We now introduce a lemma that establishes a connection between the undirected edge geography problem and the logics we have developed.

\begin{lemma}
For any rooted undirected graph $(G, d)$, Player I has a wining strategy in the undirected edge geography game on $(G,d)$, if and only if $M_{G}, d \models \phi_{G}$.
\end{lemma}
\begin{proof}
We show the lemma by induction on $|R|$. 
\underline{Base case $|R| = 0$}, $n=2$. For any $x,y\in D$, $\{x,y\} \notin R$. Player I loses in this case.
Let $M_G=(D, E, C, \beta)$ be the induced model. Then $E(x, y) = \emptyset$ for any $x, y\in D$. We need to show $M_G, d \not \models \phi_G$. For any finite non-empty $S\subseteq\sk$, consider the model $M' = (D, E, C^{a_1+S}, \beta)$.
Since $\phi_G = \phi_{2} = \diamondplus_{a_1}(\psi_{1}\land\neg\chi_{1}\land K_{a_1}\boxplus_{a_2}(\neg\psi_{2}\lor\chi_{2}))$, where $\psi_1=\neg K_{a_1}\bot \land \bigvee_{x\in D}K_{a_1} p_x$, $\chi_1 = \bot$, $\psi_2 = \neg K_{a_2}\bot \land \bigvee_{x\in D}K_{a_2} p_x$, and $\chi_2 = \bigvee_{x \neq y \in D} (p_x \wedge \hat K_{a_1} p_y \wedge K_{a_2} p_y)$. It is clear that $M', d \not \models \psi_1$, since $M', d \models K_{a_1} \bot$. It follows that $M', d \not \models \psi_{1}\land\neg\chi_{1}\land K_{a_1}\boxplus_{a_2}(\neg\psi_{2}\lor\chi_{2})$. Since $S$ is arbitrary, we have $M_G, d \not\models\phi_G$.

\underline{Base case $|R| = 1$}, and so $n=2$.  Let $\{d, d'\}$ be the unique edge in $R$. Let $M_G=(D,E,C, \beta )$ be the induced model. $E(d, d') = E(d', d) = \{s_{\{d, d'\}}\}$, and $E(x, y)=\emptyset$ otherwise. Player I has a winning strategy in this case, and we show that $M_G, d \models \phi_G$. Consider $S=\{s_{\{d, d'\}}\}$. Let $M' = (D, E, C^{a_1+S}, \beta)$, with
$\phi_G = \phi_{2} = \diamondplus_{a_1}(\psi_{1}\land\neg\chi_{1}\land K_{a_1}\boxplus_{a_2}(\neg\psi_{2}\lor\chi_{2}))$, where:
\begin{itemize}
\item $\psi_1 = \neg K_{a_1}\bot \land \bigvee_{x\in D}K_{a_1} p_x$
\hfill ($M', d \models \psi_1$, for $M', d \models \neg K_{a_1} \bot \wedge K_{a_1} p_{d'}$)
\item $\chi_1 = \bot$
\hfill ($M', d \models \neg \chi_1$)
\item $\psi_2 = \neg K_{a_2} \bot \land \bigvee_{x\in D}K_{a_2} p_x$
\item $\chi_2 = (p_d \land \hat{K}_{a_1} p_{d'} \land K_{a_2} p_{d'}) \vee (p_{d'} \wedge \hat{K}_{a_1} p_{d} \land K_{a_2} p_{d}) \vee \bigvee_{x \neq y \in D \setminus \{d, d'\}} (p_x \wedge \hat K_{a_1} p_y \wedge K_{a_2} p_y)$.
\end{itemize}
For any finite nonempty $S' \subseteq \sk$, let $M'' = (D, E, (C^{a_1+S})^{a_2+S'}, \beta)$, we have one of the following cases:
\begin{enumerate}[label={(\arabic*)}]
\item $S' \not\subseteq S$, then $\forall x \in D$, $(C^{a_1+S})^{a_2+S'}(a_2)\nsubseteq E(d,x)$, hence $M'', d' \models \neg \psi_2$, for $M'', d'\models K_{a_2}\bot$.
\item $S'\subseteq S$, then $M'', d' \models p_{d'} \wedge \hat K_{a_1} p_{d} \wedge K_{a_2} p_d$. Thus, $M'', d'\models \chi_2$ for its right disjunct is satisfied. 
\end{enumerate}
In both case $M'', d' \models \neg \psi_2 \vee \chi_2$, and so $M', d' \models \boxplus_{a_2} (\neg \psi_2 \vee \chi_2)$, and $M', d \models K_{a_2} \boxplus_{a_2} (\neg \psi_2 \vee \chi_2)$. Together with the verifications above, we have $M_G, d \models \phi_G$.

\underline{The case $|R| = k$}. The direction from left to right. Suppose that Player I has a winning strategy, by which she chooses in the first move $\{d, d'\}$. Let $M_G = (D, E, C, \beta)$ be the induced model. We need to show that $M_G, d \models \phi_G$,
where $\phi_G= \diamondplus_{a_1}(\psi_1 \land \neg \chi_1 \land K_{a_1}\phi_{G,\boxplus_{a_2}})$, in which $\phi_{G,\boxplus_{a_2}}$ is the subformula of $\phi_G$ beginning with $\boxplus_{a_2}$ (see Def.~\ref{def:indf}). Consider $S=\{s_{\{d,d'\}}\}$, and let $M' = (D, E, C^{a_1+S}, \beta)$:
\begin{itemize}
\item $\psi_1=\neg K_{a_1}\bot \land \bigvee_{x\in D}K_{a_1}p_x$
\hfill ($M', d \models \psi_1$, for $M',d\models \neg K_{a_1}\bot \land K_{a_1} p_{d'}$)
\item $\chi_1=\bot$
\hfill ($M', d \models \neg \chi_1$)
\end{itemize}
Now we show $M',d\models K_{a_1}\phi_{G,\boxplus_{a_2}}$; namely, $M',d'\models \phi_{G,\boxplus_{a_2}}$, where $\phi_{G,\boxplus_{a_2}}=\boxplus_{a_2}(\neg \psi_2 \lor \chi_2 \lor \hat{K}_{a_2}\phi_{G,\diamondplussc_{a_3}})$ in which $\phi_{G,\diamondplussc_{a_3}}$ is the subformula of $\phi_G$ beginning wtih $\diamondplus_{a_3}$. For any finite nonempty $S'\subseteq\sk$, let $M'' = (D, E, (C^{a_1+S})^{a_2+S'}, \beta)$, and it suffices to show that 
\begin{align}\tag{\dag}
M'', d' \models \neg \psi_2 \lor \chi_2 \lor \hat{K}_{a_2}\phi_{G,\diamondplussc_{a_3}},
\end{align}
where $\psi_2=\neg K_{a_2}\bot \land \bigvee_{x\in D}K_{a_2}p_x$ and $\chi_2 = \bigvee_{x\neq y\in D}  (p_x \land \hat{K}_{a_1} p_y \land K_{a_2} p_y)$. Consider the possible cases:
\begin{enumerate}[label=(\arabic*)]
\item There does not exist $x \in D$ such that $S' \subseteq E(d',x)$, or
\item There exists $d'' \in D$ such that $S' \subseteq E(d', d'')$ (note that $S'$ must be singleton).
\end{enumerate}
In case (1), $M'',d'\models K_{a_2}\bot$, so $M'',d'\models \neg\psi_2$, hence $(\dag)$ holds. In case (2), 
Player I has a winning strategy in the continued game on $(G_2, d'')$ with $G_2=(D, R \setminus \{ \{d, d'\}, \{d',d''\} \})$ (note that $d''$ cannot be $d$ or $d'$). It suffices to show the following result:
\begin{align}\tag{\ddag}
M'',d''\models \phi_{G,\diamondplussc_{a_3}} \iff M_{G_2},d''\models \phi_{G_2}.
\end{align}
Since $M_{G_2}, d'' \models \phi_{G_2}$ holds by the induction hypothesis, by $(\ddag)$, we have $M'',d''\models \phi_{G,\diamondplussc_{a_3}}$. This makes the rightmost disjunct of $(\dag)$ true in $M'', d'$, and completes the whole proof.

Let $M_{G_2}=(D,E_2,C,\beta)$. To see $(\ddag)$, $M'',d''\models \phi_{G,\diamondplussc_{a_3}}$, i.e., $(D, E, (C^{a_1+S})^{a_2+S'}, \beta),d''\models \phi_{G,\diamondplussc_{a_3}}$
\begin{enumerate}
\item[$\Longleftrightarrow$] 
$(D,E_2,(C^{a_1+S})^{a_2+S'},\beta),d''\models \phi'_{G,\diamondplussc_{a_3}}$, where $\phi'_{G,\diamondplussc_{a_3}}$ is adapted from $\phi_{G,\diamondplussc_{a_3}}$ by the following:
\begin{itemize}
\item Delete all occurrences of $\bigvee_{x\neq y\in D}  (p_x \land \hat{K}_{a_1} p_y \land K_{a_i} p_y)$ from $\phi_{G,\diamondplussc_{a_3}}$
\item Delete all occurrences of $\bigvee_{x\neq y\in D}  (p_x \land \hat{K}_{a_2} p_y \land K_{a_i} p_y)$ from $\phi_{G,\diamondplussc_{a_3}}$
\end{itemize}
(This equivalence holds since $E(d, d') = E(d', d) = \emptyset$, which implies that any formulas $\hat K_{a_1} \phi$ and $\hat K_{a_2} \phi$ are false in any world $x$ of model $(D,E_2,C',\beta)$, where $C'$ is any capability function updated from $(C^{a_1+S})^{a_2+S'}$ without changing the capabilities of $a_1$ and $a_2$.)
\item[$\Longleftrightarrow$] 
$(D,E_2,C,\beta),d''\models \phi''_{G,\diamondplussc_{a_3}}$, where $\phi''_{G,\diamondplussc_{a_3}}$ a variant of $\phi'_{G,\diamondplussc_{a_3}}$ by replacing any $a_{i+2}$ with $a_i$,\\
 (This holds since $(C^{a_1+S})^{a_2+S_2}(a_{i+2})=C(a_i)=\emptyset$; note that $a_1$ and $a_2$ does not exist in $\phi'_{G,\diamondplussc_{a_3}}$.)
\item[$\Longleftrightarrow$] 
$M_{G_2},d''\models\phi_{G_2}$, i.e.,
$(D,E_2,C,\beta),d''\models \phi_{G_2}$
\hfill (since $\phi_{G_2}=\phi''_{G,\diamondplussc_{a_3}}$)
\end{enumerate}

From right to left. If Player I does not have a winning strategy, we must show that $M_G, d \not\models \phi_G$. Let the induced model $M_G$ be $(D,E,C,\beta)$. Since Player I does not have a winning strategy, then:
\begin{enumerate}[label=(\alph*)]
\item There is no $x\in D$ such that $\{d, x\} \in R$, and Player I loses in this case; or
\item Player I does not have a winning strategy by choosing in the first move any $x \in D\setminus \{d\}$ such that $\{d, x\} \in R$.
\end{enumerate}

For case (a): Since $E(d, x)=\emptyset$ for any $x$, we get $M_G, d \not \models \phi_G$ similarly to the case when $|R| = 0$.

For case (b): 
Consider an arbitrary finite nonempty $S \subseteq \sk$. Then:
\begin{enumerate}[label=(\arabic*)]
\item For all $x\in D$, $S\nsubseteq E(d,x)$; or
\item Theres exists $d' \in D$ such that $S \subseteq E(d, d')$ (note that $d'$ cannot be $d$).
\end{enumerate}
We need to show $M_G, d \not \models \phi_G$ where $\phi_G$ is given in Def.~\ref{def:indf}. Let $M' = (D, E, C^{a_1+S}, \beta)$. In case (1), since $M' ,d \models K_{a_1} \bot$, $M', d \not \models \psi_1$ (with $\psi_1=\neg K_{a_1}\bot \land \bigvee_{x\in D}K_{a_1}p_x$), and so $M ,d \not\models \phi_G$. 

In case (2) (under the case (b)), there must exist $d''\in D\setminus\{d,d'\}$ such that Player I does not have a winning strategy in the game on $(G_2, d'')$ where $G_2 = (D, R \setminus \{ \{d, d'\}, \{d',d''\} \})$; for otherwise Player I has a winning strategy (this is also the case when there is no such a $d''$), leading to a contradiction. Let $S'=\{s_{\{d', d''\}} \}$, then $S'\subseteq E(d',d'')$. Let $M'' = (D, E, (C^{a_1+S})^{a_2+S'}, \beta)$. It suffices to show that
\begin{align}\tag{*}
M'', d' \not \models \neg \psi_2 \lor \chi_2 \lor \hat{K}_{a_2}\phi_{G,\diamondplussc_{a_3}},
 \end{align}
Consider $\psi_2=\neg K_{a_2}\bot \land \bigvee_{x\in D}K_{a_2}p_x$. Since $M'',d'\models \neg K_{a_2}\bot \land K_{a_2} p_{d''}$, we have $M'',d' \not\models \neg \psi_2$.
As for $\chi_2 = \bigvee_{x\neq y\in D}  (p_x \land \hat{K}_{a_1} p_y \land K_{a_2} p_y)$, since $M'',d'\models \hat{K}_{a_1}p_y \land K_{a_2} p_y$ implies $y=d\neq d''=y$, we have $M'',d' \not\models \chi_2$.
Finally we show that $M'', d' \not \models \hat{K}_{a_2}\phi_{G,\diamondplussc_{a_3}}$. Since there is exact one $x \in D$ (which must be $d''$ by the definition of $S'$) such that $S' \subseteq E(d',x)$, it suffices to prove $M'',d'' \not\models \phi_{G,\diamondplussc_{a_3}}$. Note that $(\ddag)$ from the proof of the converse direction can also be shown here, it suffices to show that $M_{G_2}, d'' \not\models \phi_{G_2}$, and this holds by the induction hypothesis.
\end{proof}

\begin{corollary}\label{lem:red-gg2cua}
Undirected edge geography is polynomial time reducible to the model checking problem for $\l_{\boxplus}$.
\end{corollary}

\begin{remark}
It is important to note that the reduction discussed previously utilizes only the modalities $\boxplus$ and $\diamondplus$. However, we can also perform a reduction using exclusively the modalities $\Box$ and $\Diamond$. This alternative reduction is structurally similar to the original, with the primary modification being the replacement of $\boxplus$ with $\Box$. Additionally, a reduction that employs only the modalities $\boxminus$ and $\diamondminus$ is also feasible. In this case, we replace $\boxplus$ with $\boxminus$. Furthermore, there is a requirement to modify the skill set $C(a_i)$ to $\{ s_{\{w, v\}} \mid w, v \in D \}$. The model checking problems for any logics that include at least one of the modalities -- $\boxplus$, $\boxminus$, $\Box$, $\diamondplus$, $\diamondminus$ or $\Diamond$ -- remain PSPACE hard. This complexity assertion holds even in the absence of additional modalities such as $C_G$, $D_G$, $E_G$, $F_G$, $(+_S)_a$, $(-_S)_a$, $(=_S)_a$, and $(\equiv_S)_a$.
\end{remark}

\begin{lemma}\label{thm:complexity-cua}
The model checking problem for $\l_{CDEF+-=\equiv\boxplus\boxminus\Box}$ is in PSPACE.
\end{lemma}
\begin{proof}
With the presence of Algorithm~\ref{alg:val-dcx}, it suffices to provide a polynomial space algorithm for the types of formulas $\boxplus_a\phi$, $\Box_a\phi$ and $\boxminus_a\phi$. The details are given in Algorithm~\ref{alg:val-box}.

\begin{algorithm}
\caption{Function $Val((W,E,C,\beta),\phi)$ extended: cases with quantifiers}\label{alg:val-box}
\includegraphics{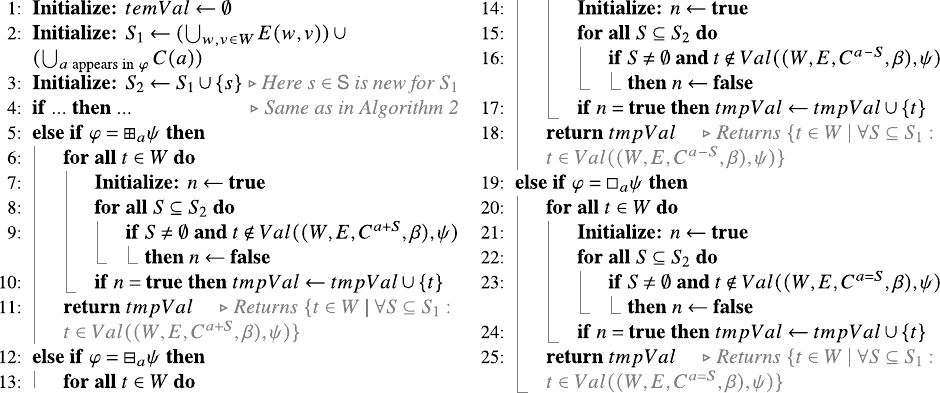}
\end{algorithm}

Here we furthermore need to check the space cost caused by the new modalities. But notice that all the space cost of checking $Val((W,E,C,\beta),\phi)$ is in $O(|M|)$. So the space cost of the algorithm is immediately linear.  So the model checking problem for $\l_{CDEF+-=\equiv\boxplus\boxminus\Box}$ is in PSPACE.
\end{proof}

We reach the following result from Corollary~\ref{lem:red-gg2cua} (considering that UEG is PSPACE complete) and Lemma~\ref{thm:complexity-cua}.

\begin{theorem}\label{lem:mc-lcua}
The model checking problems for all logics with quantifiers (i.e., at least one of $\boxplus$, $\boxminus$ and $\Box$) that extends the base logic \l is PSPACE complete.
\end{theorem}

\section{Discussion}
\label{sec:conclusion}

We have developed a variety of logics that incorporate individual and group knowledge, actions such as knowing, forgetting, revising, and learning, as well as the necessity and possibility of these actions. These logics are highly expressive, yet the computational cost for model checking remains manageable. Specifically:
\begin{itemize}
\item For logics devoid of quantifiers, the complexity of model checking falls within the class P, aligning with many traditional epistemic logics.
\item For logics that include quantifiers, the complexity is PSPACE complete. This matches the complexity found in similar types of logics, such as Group Announcement Logic \cite{ABDS2010}, Coalition Announcement Logic \cite{Pauly2002,GAD2018,ADGW2021}, and Subset Space Arbitrary Announcement Logic \cite{BDK2013}.
\footnote{It is worth noting that model checking in Arbitrary Public Announcement Logic is also believed to be PSPACE complete \cite{BBDHHL2008}. However, a detailed validation of this claim has not yet been found by us.}
\end{itemize}

Logicians are deeply interested in the decidability of validity/satisfiability problems for logics that include quantifiers over updates. Known complexities, such as the undecidability of Arbitrary Public Announcement Logic (APAL) and Group Announcement Logic \cite{FD2008, ADF2016}, have spurred ongoing research into decidable alternatives \cite{FD2008, DFP2010, DF2022}. Even the development of variants that are recursively axiomatizable represents an  advancement \cite{XW2018, BOS2023}, especially given that APAL is not likely to have this feature.

Our research also aims to explore the decidability and computational complexity of satisfiability and validity problems within our logics. Although our ongoing efforts have yielded PSPACE completeness and EXPTIME completeness results for many of our less complex logics (for example, the satisfiability problems for logics devoid of  common knowledge, update modalities, and quantifiers are PSPACE complete, whereas those lacking  update modalities and quantifiers but incorporating common knowledge are EXPTIME complete), a definitive result for the full logic $\l_{CDEF+-=\equiv\boxplus\boxminus\Box}$ remains elusive. Additionally, while we have successfully axiomatized some of our logics in previous studies \cite{LW2022b}, an axiomatic system for the full logic is not yet developed. These areas are designated for future exploration and development.

We have introduced a new update modality for learning, $(\equiv_b)_a$, which denotes the action where agent $a$ learns the skills of agent $b$. This operator essentially replaces $a$'s skill set with that of $b$. However, we can also devise variants that facilitate skill set modification through incremental learning (e.g., $(\cup_b)_a$) or decremental learning (e.g., $(\cap_b)_a$ ``retaining only beneficial skills from $b$'', or $(\setminus_b)_a$ ``eliminating undesirable skills of $b$''). Additionally, the concept of ``deskilling,'' derived from the richness of natural language, refers to a reduction in the skills required to perform a task. This could be modeled as an update action that modifies the edge function, thereby requiring fewer skills to distinguish between worlds, potentially leading to knowledge acquisition. Incorporating these diverse learning modalities does not increase the complexity of the model checking problem, though it may add complexity to the validity problem. The exploration of quantifiers over learning operators presents another intriguing area of study.

\paragraph{Acknowledgements}
We express our gratitude to the anonymous reviewers for their invaluable comments and suggestions. We acknowledge the financial support by the MOE Project of Humanities and Social Sciences (No.~24YJA72040002) and the National Social Science Fund of China (Grant No.~20\&ZD047).

\bibliographystyle{eptcs}
\bibliography{generic}
\end{document}